\documentclass[unicode,orivec]{llncs}

\usepackage{graphicx}
\graphicspath{{./graphics/}}

\usepackage{hyperref}
\usepackage{microtype}
\usepackage[table]{xcolor}
\usepackage{adjustbox}
\usepackage[size=footnotesize]{todonotes}
\usepackage{xspace}
\usepackage{complexity}
\usepackage{cite}
\usepackage{booktabs}
\usepackage{optidef}
\usepackage{longtable}
\usepackage{thm-restate}
\usepackage[capitalize]{cleveref}
\usepackage{amssymb}

\usepackage{amsthm}
\usepackage[shortlabels]{enumitem}
\usepackage{subcaption}
\usepackage{wrapfig}
\usepackage{lineno}
\usepackage{tikz}
\usepackage{pgfplots}
\usetikzlibrary{shapes}
\usetikzlibrary{plotmarks}
\usepackage[linesnumbered,algoruled,longend,vlined]{algorithm2e}
\DontPrintSemicolon
\SetArgSty{textnormal}
\SetProgSty{textnormal}
\SetFuncArgSty{textnormal}
\SetProcNameSty{textnormal}
\SetProcArgSty{textnormal}
\SetKwInput{Input}{Input}
\SetKwInput{Output}{Output}

\usepackage{apptools}

\newcommand\lncs[1]{}
\newcommand\arxiv[1]{#1}

\arxiv{\newcommand{\restateref}[1]{\IfAppendix{\hyperref[#1]{$\star$}}{\hyperref[#1*]{$\star$}}}}
\lncs{\newcommand{\restateref}[1]{\href{https://arxiv.org/abs/2607.06029}{$\star$}}}

\spnewtheorem{observation}{Observation}{\bfseries}{\normalfont}
\crefname{observation}{observation}{observations}
\Crefname{observation}{Observation}{Observations}
\crefname{problem}{Problem}{Problems}
\Crefname{problem}{Problem}{Problems}
\newcommand{\probname}[1]{{\normalfont\textsc{#1}}\xspace}
\newcommand{\CChain}{bucket arrangement\xspace} % why?

\newcommand{\BW}{\probname{Bucketwidth}}
\newcommand{\BI}{\probname{Bucket Integrity}}
\newcommand{\OBI}{\probname{Ordered Bucket Integrity}}

\newcommand{\dmin}{\ensuremath{\delta_{\mathrm{min}}}\xspace}
\newcommand{\true}{\ensuremath{\text{\textsc{true}}}\xspace}
\newcommand{\false}{\ensuremath{\text{\textsc{false}}}\xspace}
\newcommand{\tfspaths}{$\{3,5,7\}$-paths\xspace}

\newcommand{\crect}[1]{\tikz{\fill[#1] (0,0) rectangle (1ex,1.5ex);}}
\newcommand{\fakemarko}[1]{\tikz{\draw[#1,thick] (0,0) circle (0.5ex);}}
\newcommand{\fakemarkp}[1]{\tikz{\draw[#1,thick] (0,0.5ex) -- (1ex,0.5ex) (0.5ex,0) -- (0.5ex,1ex);}}
\newcommand{\fakemarkx}[1]{\tikz{\draw[#1,thick] (0,0) -- (1ex,1ex) (0,1ex) -- (1ex,0);}}
\newcommand{\fakeareas}[1]{\tikz{\draw[#1,thick] (0,0) rectangle (2ex,1.33ex);}}

\newcommand{\fakelines}[1]{\tikz{\draw[#1,thick] (0,0) -- (2ex,1.33ex);}}
\newcommand{\fakelined}[1]{\tikz{\draw[#1,thick,densely dotted] (0,0) -- (2ex,1.33ex);}}

\newcounter{tkeycnt}
\makeatletter
\newcommand{\tkey}[1]{%
    \refstepcounter{tkeycnt}%
    \def\@currentlabel{#1}%
    \def\cref@currentlabel{[tkey][\the\value{tkeycnt}][131072]#1}%
    \label{tkey:#1}%
    \texttt{(#1)}}
\makeatother

\crefformat{tkey}{#2\texttt{(#1)}#3}
\crefrangeformat{tkey}{\errmessage{formatting as range not supported for tkey}}
\crefmultiformat{tkey}{#2\texttt{(#1)}#3}{ and #2\texttt{(#1)}#3}
    {, #2\texttt{(#1)}#3}{ and #2\texttt{(#1)}#3}

\newcolumntype{R}[2]{%
    >{\adjustbox{angle=#1,lap=\width-(#2)}\bgroup}%
    l%
    <{\egroup}%
}

\newclass{\XNLP}{XNLP}

\definecolor{PKdarkred}{rgb}{0.89 0.102 0.109}
\definecolor{PKdarkblue}{rgb}{0.121 0.47 0.705}
\definecolor{PKdarkgreen}{rgb}{0.2 0.627 0.172}
\definecolor{PKdarkorange}{rgb}{1 0.498 0}
\definecolor{PKdarkpurple}{rgb}{0.415 0.239 0.603}
\definecolor{PKdarkyellow}{rgb}{1 1 0.2}
\definecolor{PKdarkbrown}{rgb}{0.651 0.337 0.157}
\definecolor{PKdarkpink}{rgb}{0.969 0.506 0.749}
\definecolor{PKdarkcyan}{rgb}{0.106 0.62 0.467}
\definecolor{PKdarkgray}{rgb}{0.5 0.5 0.5}
\definecolor{PKlightred}{rgb}{0.984 0.603 0.6}
\definecolor{PKlightblue}{rgb}{0.651 0.807 0.89}
\definecolor{PKlightgreen}{rgb}{0.698 0.874 0.541}
\definecolor{PKlightorange}{rgb}{0.992 0.749 0.435}
\definecolor{PKlightpurple}{rgb}{0.792 0.698 0.839}
\definecolor{PKlightyellow}{rgb}{1 1 0.8}
\definecolor{PKlightbrown}{rgb}{0.898 0.847 0.741}
\definecolor{PKlightpink}{rgb}{0.992 0.855 0.925}
\definecolor{PKlightcyan}{rgb}{0.553 0.827 0.78}
\definecolor{PKlightgray}{rgb}{0.8 0.8 0.8}

\definecolor{DEUred}{rgb}{1 0 0}
\definecolor{USAblue}{rgb}{0.098 0.133 0.251}

\definecolor{dark blue}{rgb}{0.121,0.47,0.705}
\let\emph\relax\DeclareTextFontCommand{\emph}{\color{dark blue}\em}

\renewcommand{\orcidID}[1]{\href{https://orcid.org/#1}{\includegraphics[scale=.03]{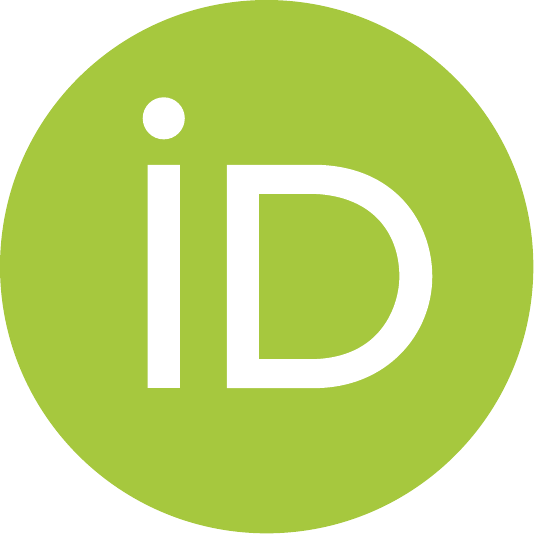}}}

\begin{document}
\title{A Combinatorial Framework for \\ Narrow Graph Drawings}

\author{
    Tim Hegemann\,\orcidID{0009-0008-4770-3391} \and 
    Dominik Jilg\,\orcidID{0009-0004-4527-9583} \and
    \mbox{Marie Diana Sieper}\,\orcidID{0009-0003-7491-2811} \and
    \mbox{Samuel Wolf}\,\orcidID{0009-0009-7098-6147}
}

\institute{%
    Universität Würzburg, Germany,
    \email{firstname.lastname@uni-wuerzburg.de}
}

\authorrunning{T. Hegemann, D. Jilg, M. D. Sieper, and S. Wolf} %TODO mandatory. First: Use abbreviated first/middle names. Second (only in severe cases): Use first author plus 'et al.'

%\Copyright{Tim Hegemann, Dominik Jilg, Marie Diana Sieper, Samuel Wolf} %TODO mandatory, please use full first names. LIPIcs license is "CC-BY";  http://creativecommons.org/licenses/by/3.0/

%\ccsdesc[500]{Theory of computation~Graph algorithms analysis}
%\ccsdesc[500]{Human-centered computing~Graph drawing440s}

\maketitle

\begin{figure}[h]
    \centering
    \includegraphics[page=2,width=\textwidth - 2cm]{starry-night.pdf}
    \caption{Chunky Chain visualization of spatial color co-occurrences in
            a painting. %Vincent van Gogh's “Starry Night Over the Rhône”.  
            When turned 90°, the drawing fits the narrow screen of a smartphone.}
    \label{fig:teaser}
\end{figure}

\begin{abstract}
    We introduce a combinatorial framework for narrow graph drawings via
    \emph{bucket arrangements} designed for small screens such as smartphones, 
    where vertical scrolling is the dominant means of interaction.
    Given a \emph{capacity}~$c$, the vertices are partitioned into an ordered
    set of buckets, each of size at most~$c$.
    Edges whose endpoints lie in the same or in adjacent buckets are \emph{short}.
    Edges that are “skipping” at least one bucket are \emph{long}, and we draw
    them only partially. The goal is to minimize the number of long edges.

    We develop exact and heuristic algorithms for computing bucket arrangements, 
    and experimentally evaluate their effectiveness.
    We devise an implementation of our framework, called \emph{Chunky Chains},
    that draws each bucket as a chord diagram. 
    Our experiments show that many real-world graphs have good Chunky Chain 
    visualizations.

    Find our implementation at \href{https://github.com/hegetim/chunky-chains}{github.com/hegetim/chunky-chains}.


\end{abstract}

\section{Introduction}

Smartphone screens are arguably the most common type of screen that
displays information nowadays.
Due to their limited size, 
users explore content through zooming, panning, and (vertical) scrolling.
In many applications, such as social media feeds, scrolling is the dominant
mode of navigation.  Following this example, we explore graph visualizations
with small width (but possibly large height).

Such visualizations naturally reflect an underlying structure of the data for
certain classes of graphs.
For instance, temporal graphs often have an inherent ordering in time,
and arranging vertices or events along a narrow axis can make this
ordering explicit.  Similar considerations arise for graphs representing
itineraries, processes, hierarchies, sequences, or the like.
Interestingly, in \cref{fig:teaser} our algorithm unveiled such a structure
that we previously were not aware of.  The figure visualizes the spatial color
co-occurrences in Vincent van Gogh's “Starry Night Over the Rhône” (see
\arxiv{\cref{apx:colorgraphs}}\lncs{the full version~\cite{hjsw-cc-arxiv}} for
details).  Our algorithm apparently ordered
the colors from blue, over teal and green, to ocher.

An intuitive approach to handling narrow screens
is to use {\em arc diagrams} (see e.g.\ \cite{dfgs-llr-esa25}).
In arc diagrams, vertices are arranged along a line and edges are drawn
using convex curves at either side of that line.
A common style of arc diagrams uses semicircles for the edges.
To make these drawings as narrow as possible, the ordering of the vertices 
should be one that minimizes the bandwidth of the graph.
This approach has two drawbacks. Firstly, arranging all vertices
on a line draws focus to the ordering of the vertices while possibly obscuring
other properties. Secondly, the width of the drawing is inherently 
tied to the bandwidth of the graph. Thus, only graphs admitting a low-bandwidth 
ordering can be represented effectively in an extremely narrow format.
We retain the idea of arranging vertices along a single axis but relax both
restrictions. Instead of assigning every vertex a distinct position
on the line, we group vertices into a chain of \emph{buckets}, where each bucket
may be represented in a restricted area which provides an additional degree of
freedom to visualize (local) properties apart from the linear structure.
Combinatorially, we want to find a \emph{bucket arrangement} of a given
graph~$G$.  A bucket arrangement is a partition of $V(G)$ into buckets 
$B_1,\dots,B_k$. 
Ideally, we want to find a bucket arrangement such that every edge
is \emph{short}, that is, $|i-j|<2$ for every edge~$\{u,v\}$ with
$u\in B_i, v\in B_j$.
The smallest \emph{width} $\max_i |B_i|$ of such a bucket arrangement
is a graph parameter known as \emph{bucketwidth} in the literature
(see e.g.\ \cite{dfu-hrab-jcss11}). 

\begin{problem}[\BW]\label{def:bw}
    Given a graph $G$, find a bucket arrangement of~$G$ of minimum width that
    consists of only short edges.
\end{problem}

Unfortunately,
it can be shown that the bucketwidth of a graph is within
a factor of 2 of its bandwidth~\cite{dfu-hrab-jcss11}. 
In order to draw a broader range of graphs narrowly, 
we introduce a parameter~$c$, called \emph{capacity}, which limits the
number of vertices in each bucket. Under this constraint it is possible that
not every edge is short. We call edges that skip at least one bucket \emph{long}. 
Since long edges may require extra space to draw, we want to avoid them.
This gives rise to the following combinatorial optimization problem which 
constitutes the core of our framework.

\begin{problem}[\BI]\label{def:bi}
    Given a graph $G$ and a capacity~$c\in\mathbb{N}$,
    find a bucket arrangement of~$G$ of width at most~$c$
    and the minimum number of long edges.
\end{problem}

Note that this problem can be viewed as an edge-deletion distance problem
to the decision version of \BW. To ensure that the width of our drawing
is bounded by a function in $c$, we do not require that long edges
are drawn completely.

In the following, we set $n=|V(G)|$ and $m=|E(G)|$ if the graph~$G$ is 
clear from the context, and we write $N(v)$ for the neighborhood of a
vertex~$v$.
For a given bucket arrangement, $\mathcal{L}$ denotes the set of
long edges, and we set $\ell = |\mathcal{L}|$.

\subsubsection*{Contribution and Organization.}
We introduce a new combinatorial framework for drawing graphs on narrow
surfaces via \BI.
We investigate theoretical aspects of our combinatorial
framework and practical approaches.
In particular, we analyze the (parameterized) complexity of \BI
and propose an \FPT-algorithm with respect to the vertex cover number,
and a single-exponential time exact algorithm.
Furthermore, we consider a tractable variant of \BI, called \OBI,
motivated by heuristics (see \cref{sec:combi}). 

Our framework engenders a family of drawing styles. We introduce one 
implementation, called \emph{Chunky Chains}, 
which realizes the bucket arrangement with a series of chord diagrams
(see~\cref{sec:pipeline}).
The pipeline for drawing Chunky Chains consists of three steps.
Firstly, vertices are assigned to buckets.  Here we devise an ILP and heuristics.
Secondly, we decide which edges may run {\em confluently} in order to make our
drawings even narrower.
Finally, we compute an ordering of the vertices within the buckets to
minimize crossings.
We conclude with
an experimental evaluation (\cref{sec:experiments}) and open questions 
(\cref{sec:conclusion}).

\subsubsection*{Related Work.}  
While an aspect ratio (width/height) close to 1 is often pursued in graph drawing,
technical and environmental constraints can push towards an unbalanced
aspect ratio even apart from digital screens. A historic example is the 
{\em Tabula Peutingeriana}~\cite{tabula-peutingeriana}, a Late Antique Roman road
map on an enormous parchment scroll measuring 33\,cm in height and 680\,cm in
length. % (a stunning aspect ratio of approximately 0.05).

\BI is a generalization of
\textsc{Bucketwidth} (see \cref{def:bw,def:bi}).
\BW has mainly been investigated as a means to obtain approximation results
for \textsc{Bandwidth} since the bucketwidth of a graph~$G$ is
within a factor of~2 of the bandwidth of~$G$~\cite{u-cabp-focs98,dfu-hrab-jcss11,%
ft-abc-approx05,ft-atboc-algo09,fgk-et2ab-ipec09,fgk-aet2aafb-tcs13}.
Aside from applications for \textsc{Bandwidth}, \BW is
related to several width measures in parameterized complexity
theory. In particular, it is the counterpart to tree-partition
width~\cite{w-otpw-eujc09,bgj-pcctp-ipec22,bgj-otpcoctp-dmtcs25}
in the same way as pathwidth is the counterpart to treewidth.
Tree-partition width is defined analogously to \BW but the
quotient graph of the partition is allowed to be a tree.
Several intractable problems for treewidth become tractable for
tree-partition width~\cite{bcw-phftbefsg-wg22}.
Partition width measures are also important in graph drawing.
For example, 
for drawing graphs with few slopes~\cite{dsw-gdwfs-jocg07}, computing straight
line 3d grid graphs with small volume~\cite{dm-tdgcqn-gd03,dlm-csl3dgdogilv-jocg05},
and bounding several notions of crossing numbers~\cite{wt-pdatcnogwaem-gd07}.

We are not the first ones to consider graph drawings on smartphone screens.
Da Lozzo, Di Battista, and Ingrassia~\cite{ddi-dgs-gd10,ddi-dgs-jgaa12} draw the
neighborhood of a user-selected focus vertex and explore smartphone-specific
interaction methods.  Aulbach, Fink, Schuhmann, and Wolff~\cite{afsw-dgra-gd14}
select a maximum-weight subgraph that can be drawn inside the restricted area of
a small screen.  For planar graphs, only polynomially
bounded coordinate-precision is required to draw it entirely inside a fixed
polygon~\cite{cegl-pofpa-gd10,cegl-pofpa-jgaa12}.
Another way to tackle the issue of small screens is to limit the aspect ratio.
Efficient algorithms are known for binary trees~\cite{gr-btlaaar-gd02,gr-btlaaar-jgaa04}.
Heuristic methods have been presented for layered
drawings~\cite{resh-gdglp-gd16,resh-glafda-jgaa17} and, very recently,
orthogonal layouts~\cite{ark+-arcol-eurovis26}.

Arc diagrams can be traced back to early works on crossing minimization in
{\em linear layouts} by Saaty~\cite{saaty_minimum_1964} and
Nichelson~\cite{n-ppfmtnocian-68}.  See~\cite{dfgs-llr-esa25} for a contemporary
overview of linear layouts.  Wattenberg~\cite{wattenberg_arc_2002} popularized
the usage of arc diagrams as visualization technique when visualizing repetition
patterns in strings and music.  For a survey of applications
see Nagel and Duval~\cite{nagel_visual_vis13}.
In a user study, 139~participants performed cluster recognition tasks faster and
more accurate with both, arc and chord diagrams, compared to force-based node-link
drawings~\cite{ammg-ivsgcol-tvcg25}.

Our design of Chunky Chains is inspired by 
ChordLink~\cite{admpt-clhvm-gd19,admpt-hgvcl-tvcg20},
a hybrid drawing style combining node-link and chord diagrams~\cite{ksb+-ciacg-gr09}.  
ChordLink was introduced by Angori, Didimo, Montecchiani, Pagliuca, and Tappini.  While
ChordLink integrates into an interactive workflow and tries to keep the layout
stable, our framework is designed to create a layout from scratch without user
input.  Several steps of the ChordLink pipeline are \NP-hard to
compute~\cite{ksw-cccl-jgaa23}.  
In a user study, 82~participants performed tasks
related to structural properties of the network more accurate (albeit slower)
with ChordLink compared to classical node-link diagrams~\cite{ddmt-ushgv-gd21,ddlmt-sehvg-tvcg24}.
Our design of gate nodes, that allow us to avoid the node duplication performed
by ChordLink, is inspired by Baur and Brandes' multi-circular layout~\cite{bb-mclmmg-gd07}.

\section{Combinatorial Results}\label{sec:combi}

We investigate the (parameterized) complexity of \BI and derive
useful characteristics of bucket arrangements, starting with an
upper bound on the number of buckets for optimal bucket arrangements. 
This upper bound is a crucial ingredient in our integer linear programming (ILP)
formulation (see \cref{ssec:arrangement}). 
Later we consider a variant of \BI called \OBI which turns out to be
useful for heuristic algorithms.

Complete proofs for statements marked with ($\star$) can be found
in~\arxiv{\cref{apx:combinatorial}}\lncs{the full version~\cite{hjsw-cc-arxiv}}. 

\begin{restatable}[\restateref{lem:num-buckets-upper-bound}]{lemma}{numbucketsupperbound}
    \label{lem:num-buckets-upper-bound}
    Let $(G, c)$ be an instance of \BI and $c\leq \frac{n}{2}$.
    There is an optimal \CChain of $G$ that uses at most
    $\lfloor\frac{3n}{2c}\rfloor$ many buckets
    and $\lfloor\frac{3n}{2(c+2)}\rfloor$ buckets are sometimes required.
\end{restatable}
\begin{proof}[Proof Sketch]
    If there are more than $\lfloor\frac{3n}{2c}\rfloor$ buckets,
    there are three consecutive buckets that can be merged into two
    without introducing new long edges.
    A path of cliques of size $\lfloor 2c/3 \rfloor+ 1$ requires
    $\lfloor\frac{3n}{2(c+2)}\rfloor$ many buckets.
\end{proof}

Note that for every instance with $c > n/2$ two buckets always suffice,
making long edges unnecessary.

\subsubsection*{A (Parameterized) Complexity Overview.}
Since \BI is a generalization of \BW, we know that \BI is \NP-hard
due to the \NP-hardness of \BW. 
We propose an exact algorithm using dynamic programming that
leverages the path structure of a \BI solution.

\begin{restatable}[\restateref{thm:bi-exact-single-exp}]{theorem}{biexactsingleexp}
    \label{thm:bi-exact-single-exp}
    Let $(G, c)$ be an instance of \BI.
    An optimal solution of $(G, c)$ can be computed in
    $O^*(4.2426^n)$ time and $O^*(2.8285^n)$ space.
\end{restatable}
\begin{proof}[Proof Sketch]
    We use dynamic programming over subsets of vertices of~$G$.
    Let~$L$ be a subset of~$V(G)$ of size at most $n/2$ 
    and let $G[L]$ be the subgraph of~$G$ 
    induced by~$L$. We define a table entry~$D[L, B]$ to be the minimum
    number of long edges for a bucket arrangement of $G[L]$, where $B$ 
    is the last bucket. We can compute the table~$D$ and use it
    to compute a solution by guessing a bucket that divides an optimal
    arrangement into two halves with at most $n/2$ vertices.
\end{proof}

Due to the connection between \BI and \BW, we can consult
known results of \textsc{Bucketwidth} to derive several observations
about the (parameterized) complexity of \BI.
\begin{observation}\label{obs:bi}
    \BI is
    \begin{enumerate}[(i),leftmargin=2.5em]
        \item\label{np-hard-pw-tw}
         \NP-hard even for graphs with bounded pathwidth
         and treewidth,
        \item\label{np-hard-l}
         \NP-hard even if the number of long edges is bounded,
        \item\label{np-hard-max-degree}
         \NP-hard even if the maximum degree of the input graph
         is bounded,
        \item\label{np-hard-c}
         \NP-hard even if the capacity $c$ is bounded by a constant,
        \item\label{inapprox}
         inapproximable within any constant factor unless
        $\P = \NP$.
    \end{enumerate}
\end{observation}

Observations~\ref{np-hard-pw-tw} and~\ref{inapprox} follow
from the fact that \textsc{Bandwidth} cannot be approximated
within any constant even for caterpillars of hairlength~3 unless $\P = \NP$
and an observation that \textsc{Bucketwidth} and \textsc{Bandwidth}
are within a factor of 2 from each other~\cite{dfu-hrab-jcss11}.
Observation~\ref{np-hard-l} is a consequence of the fact that
\textsc{Bucketwidth} is \NP-hard~\cite{dfu-hrab-jcss11}.
Observation~\ref{np-hard-max-degree}
and~\ref{np-hard-c} follow from the following insight:
We can solve \textsc{Hamiltonian Path} using \BI by setting $c=1$.
Then, a given graph $G$ has a Hamiltonian path if and only if
the number of long edges is $|E(G)| - n + 1$.
However, \textsc{Hamiltonian Path} is \NP-hard even
on graphs with bounded degree \cite{pv-otptsp-ja84}.

Even though \BI is \NP-hard for graphs with bounded pathwidth, we can
still derive an upper bound on the pathwidth of the input graph in terms
of the capacity and the number of long edges.
Therefore, a solution of \BI with few long edges and small capacity implies that 
the input graph has small pathwidth. 

\begin{restatable}[\restateref{lem:pw}]{lemma}{pw}
    \label{lem:pw}
    Let $(G, c)$ be an instance of \BI and let $\ell$ be the minimum number
    of long edges required in any bucket arrangement of~$G$.
    Then, $\operatorname{pw}(G) \leq 2c + \ell - 1$.
\end{restatable}

It would be desirable to derive an upper bound on a function of~$c$
and~$\ell$ that only depends on structural graph parameters, but not
on the input size.
Unfortunately, for any combination of parameters that are invariant under
(graph) duplication such as $\operatorname{pw}(G)$ or the maximum degree of~$G$,
this is impossible if $\ell > 0$.
To see this, suppose a graph $G$ and a given capacity~$c$ requires~$\ell$
long edges. Then, we can simply take $k$ copies of $G$. 
As a consequence, the number of long edges increases to $k\cdot\ell$ but
the graph parameters do not change.

On the positive side, \BI becomes fixed-parameter tractable when
parameterized by the vertex cover number of the input graph.
\begin{restatable}[\restateref{thm:fpt-vc}]{theorem}{fptvc}
    \label{thm:fpt-vc}
    Let $(G, c)$ be an instance of \BI and let $C$ be a minimum
    vertex cover of $G$.
    An optimal solution of $(G, c)$ can be computed in
    $2^{O(|C|\log|C|)}\cdot n^{O(1)}$ time.
\end{restatable}
\begin{proof}[Proof Sketch]
    We show that an assignment of vertices in~$C$ into $3|C|$
    buckets always suffices (some are possibly empty).
    Given a bucket arrangement of~$C$, we show that placing
    the remaining vertices reduces to a matching problem.
\end{proof}

\subsubsection*{Ordered Bucket Integrity.}  
We consider a variant of \BI that is solvable in polynomial time.
When designing heuristics for \BI, this variant proved to be
effective (see \cref{ssec:arrangement,sec:experiments}).

\begin{problem}[\OBI]\label{def:obi}
    Let $\sigma = \langle v_1,\dots,v_n\rangle$ be an ordered set of
    vertices in a graph~$G$ and let $c\in\mathbb{N}$. 
    A \CChain $B_1, \dots, B_k$ of~$G$ is called \emph{ordered}
    (with respect to~$\sigma$) if for every $v_i \in B_p, v_j \in B_q$
    it holds that $i < j \implies p \leq q$. Find an ordered bucket
    arrangement of width at most $c$.
\end{problem}

Even though it would be tempting to simply divide the ordering into
chunks of size~$c$, this approach is not optimal (in fact, there are
instances in which many buckets are not completely filled).
However, we can still compute an optimal solution in polynomial time.

\begin{restatable}[\restateref{thm:obi}]{theorem}{obi}
    \label{thm:obi}
    \OBI can be solved in $O(nc^2)$ time and space.
\end{restatable}
\begin{proof}[Proof Sketch]
    We use dynamic programming.
    Let $D[i, d, f]$ be the minimum number of long edges of a bucket
    arrangement of the induced subgraph of the first~$i$ vertices 
    in~$\sigma$ where the last two buckets $B_{k-1}$, $B_k$ have
    a size of $|B_{k-1}| = d$ and $|B_{k}| = f$.
    Then, an optimal solution can be found via
    $\min_{1\leq d,f \leq c} D[n, d, f]$.
    
    We use the base case $D[0, d, f] = 0$ for all $1 \leq d, f\leq c$.
    For $i \geq 1$, we do the following.
    Since we compute an ordered bucket arrangement, $v_i$ has to
    be in the last bucket of any bucket arrangement of the
    first~$i$ vertices.
    Thus, any edge~$\{v_i, v_j\}$ in $G[\{v_1, \dots, v_i\}]$
    is a long edge in any ordered bucket arrangement whose last
    two buckets have sizes~$d$ and~$f$, respectively, if~$j$ is
    at most~$i-d-f$.
    Therefore, we can compute the recursive case as follows.
    \begin{align*}
        D[i, d, f] = \begin{cases} 
            D[i-1, d, f-1] + 
            |\{\{v_i, v_j\} \in E(G) \mid j \leq i-d-f\}| & \text{if } f > 1, \\
            \displaystyle\min_{1 \leq d' \leq c}\{D[i-1, d', d] + |\{\{v_i, v_j\} \in E(G) \mid j < i-d\}|\} &\text{else.}
        \end{cases}
    \end{align*}
    A naïve implementation of this dynamic program computes
    for each of the $O(nc^2)$ entries with $f>1$ its corresponding value
    in~$O(|N(v)|)$ time, but a faster implementation is possible.
    See~\arxiv{\cref{apx:combinatorial}}\lncs{the full version~\cite{hjsw-cc-arxiv}}
    for details.
\end{proof}

\section{Chunky Chains: A Pipeline}\label{sec:pipeline}

\begin{wrapfigure}[27]{R}{0.3\textwidth}\centering
    \vspace{-5ex}
    \includegraphics[page=2]{anatomy.pdf}
    \caption{Anatomy of a Chunky Chain.  Vertices are circular arcs (in orange).
        $\{\mathsf{b,g}\}$ is a long edge; the respective labels are gray.}
    \label{fig:method}
\end{wrapfigure}

We outline the functional requirements of a Chunky Chain and describe its structural
design. 

For a comprehensible drawing, vertices should be unobstructed, unambiguous,
and (optionally) labeled.  It is difficult to follow long edges, especially if
they make many turns.  Our edges are sequences of path segments, that is, either
a straight line segment, a circular arc, or a cubic Bézier curve.  We require
edges to be comprised of a constant number of path segments.

A Chunky Chain draws a bucket arrangement as a series of chord diagrams 
(see \cref{fig:method}).  
In a chord diagram every vertex is represented by a circular arc.  All vertices of
a bucket lie on a circle and each circle has the same radius.
Vertex labels can be attached to each vertex at the boundary of the
diagram where they do not obstruct other parts of the drawing.
An edge between two vertices in the same bucket is represented by a curve
inside the circle that connects the two arcs of the vertices.
Edges between two vertices of adjacent buckets are routed through two
\emph{gate nodes} (see the light blue nodes in \cref{fig:method})~-- 
special vertices (that we usually do not draw) at the top and bottom of each
chord diagram.  Edges that transit through gate nodes are drawn confluently
(see the purple circles in \cref{fig:method}).
Long edges are indicated using two arrows starting at the
incident vertices, respectively, and a corresponding pair of labels (see the
gray labels in \cref{fig:method}).

Given a bucket arrangement of width~$c$, our drawing style ensures that the
drawing's width is bounded by~$c$.  The width of the chord diagrams (and
therefore the width of the drawing) results from the number of nodes in
the diagram (which is at most~$c$) and the number of gate node transitions.
(see \cref{ssec:confluent}).

We use a three-phase pipeline to compute Chunky Chain visualizations.
The input is a graph~$G$ and a capacity~$c\in\mathbb{N}$.  Firstly, we
compute a bucket arrangement for~$G$ that has width at most~$c$ and few
long edges (see \cref{ssec:arrangement}).  Secondly, we decide which edges
should run confluently (see \cref{ssec:confluent}).  Thirdly, we order the
vertices in each bucket such that few edges cross (see \cref{ssec:crossings}).

Our algorithms optimize certain aspects of drawing quality, namely, few long
edges, few gate node transitions, and few edge crossings.

\subsection{Computing Bucket Arrangements}\label{ssec:arrangement}
The basic structure of a Chunky Chain is determined by a bucket arrangement.
Note that graph parameters can already hint at the number of long
edges. If $c < \lceil(\Delta+1)/3\rceil$, where~$\Delta$ is the maximum degree
of any vertex in~$V(G)$, any bucket arrangement of width~$c$ has long edges. 
When~$G$ has bandwidth~$b_G$ and $c \geq b_G$, then a bucket arrangement
with width~$c$ exists that has only short edges.

To compute a bucket arrangement, we present an integer linear programming (ILP)
formulation for \BI and a heuristic that is faster on large graphs.
We use~$[n]$ as a shorthand for $\{1, \dots, n\}$.

\subsubsection*{ILP Formulation for \BI.}
Instead of directly translating \cref{def:bi} into an ILP, we model
the corresponding maximization problem for a more compact
formulation. 
To this end, let $x_{i,v} \in \{0, 1\}$ be a variable for every 
$v\in V(G)$, $i\in[\Lambda]$ that equals~$1$ if and only if~$v$ is in bucket~$i$,
where $\Lambda$ is a suitable upper bound on the number of buckets needed;
for example the upper bound provided by~\cref{lem:num-buckets-upper-bound}.
Furthermore, let $y_{i,u,v} \in \{0, 1\}$ for every $\{u,v\}\in E(G)$ and $i\in[\Lambda]$
that equals~$1$ if and only if $\{u,v\}$ is a short edge and~$u$ is in bucket~$i$.
We maximize the sum of all $y_{i, u, v}$ while ensuring that
for each vertex $v$ at most one variable $x_{i,v}$ is set to~1.
We count short edges via constraints of the form 
$x_{i-1, v} + x_{i, v} + x_{i+1, v} \geq y_{i, u, v}$ and $x_{i,u} \geq y_{i, u, v}$
for every $\{u, v\} \in E(G)$. 
See \arxiv{\cref{sec:ilp-formulation}}\lncs{the full version~\cite{hjsw-cc-arxiv}}
for the concrete formulation.

\subsubsection*{A Meta-heuristic Approach.}  While state-of-the-art ILP solvers
usually find optimal bucket arrangements for graphs with up to 100 vertices in
less than one hour (see \cref{sec:experiments}),
heuristics can help in cases where we need quick response times or when
expensive solvers are unavailable.  From \cref{thm:obi} we know that \OBI can be
solved efficiently.  We use it to find good solutions for \BI.  There are many
ways to order a graph.  In our experiments (see \cref{sec:experiments}), we tested
bandwidth heuristics and a dimensionality reduction approach.

However, finding good orderings for instances where the minimum number of long
edges is large turned out to be difficult.
We propose the following strategy to find one. 
Intuitively, we iteratively simplify a given instance by removing edges that are
most likely long, thereby creating an instance in which the ordering heuristic
produces good solutions. Our algorithm works as follows.
(i) Given a black box ordering strategy, let $\sigma_G$ be the ordering determined for
graph~$G$.  Solve \OBI for~$G$ and~$\sigma_G$.  If $\ell=0$ stop and return~$\sigma_G$.
(ii) Find an edge that is probably long and remove it from~$G$.  To this end,
determine the set of edges~$\mathcal{C}$ that maximize $|\sigma_G(u) -
\sigma_G(v)|$.  Among the edges in $\mathcal{C}$, find the
edge~$\{u,v\}$ that maximizes $(\deg(u) - 1)\cdot(\deg(v) - 1)$ and remove it from~$G$.
(iii) Let~$G'$ be the graph without $\{u,v\}$.  If~$G'$ is connected, replace~$G$
with~$G'$ and go to (i).  Otherwise, recurse on the connected components of~$G'$
and concatenate the results.

\subsection{Handling Gate Node Transitions}\label{ssec:confluent}

The width of a Chunky Chain is determined by the number of vertices in any
bucket and the number of edges transiting through any gate node.
With a bucket arrangement in place, we apply a post-processing step that
reduces the number of gate transitions without increasing the width of the
chosen arrangement.

We consider all pairs of nodes~$u,v$ with~$u\in B_i$ and $v\in B_{i+1}$ such
that $B_{i-1}\cap N(u) = \varnothing$ and $B_{i+2}\cap N(v) = \varnothing$.
We calculate how many gate transitions we can avoid by
swapping~$u$ and~$v$ and greedily perform the most beneficial swap.  We repeat
this strategy until no more beneficial swaps are possible. 

We use {\em confluent} edges to further reduce the number of gate transitions~--
and therefore the width of the drawing.  We state a simplified definition of
confluent graph drawings;  a rigorous definition can be found, e.g.,
in~\cite{n-ec-acm89,degm-cd-gd03,degm-cd-jgaa05}.  Let~$G'_i$ be the subgraph
of~$G$ with only edges connecting vertices in~$B_i$ with vertices in~$B_{i+1}$
(without isolated vertices), we assign edges to the same \emph{track} when their
endpoints share the same neighborhood in~$G'_i$.  Observe that~$G'_i$ is
bipartite.  In bipartite confluent drawings, the number of tracks equals the
minimum size of a biclique edge cover of~$G'_i$~\cite{egm-cld-gd04,egm-cld-alg07}.
Note that the size of the smaller bipartition bounds the size of the biclique
cover and, therefore,~$c$ tracks are always sufficient.  Finding a biclique
cover of minimum size, however, is \NP-hard, even for bipartite
graphs~\cite{o-cgwc-im77}.
The subgraphs of only edges spanning adjacent buckets are small.  Therefore, we
can find minimum biclique covers for them in reasonable time using an ILP
formulation by Burin, Validi, Lyu, and Hicks~\cite{bvlh-ilpmbc-26}.

\subsection{Crossing Minimization}\label{ssec:crossings}

The bucket arrangement determines which vertex is part of which chord diagram.
The final phase of our pipeline is concerned with minimizing crossings %\cref{tkey:M3}
in each chord diagram by ordering the vertices and placing the gate nodes (i.e., the
sections that connect adjacent chord diagrams).  Unfortunately, crossing minimization
in circular layouts is \NP-hard~\cite{mknf-npcnlp-iscas87}. 

We model the vertex order of a bucket arrangement and use a Max-Sat
solver to
maximize the number of uncrossed edge pairs.
Let $B_1, \dots, B_k$ be a bucket arrangement.  We denote with $E_i$ the
set of edges that are between two vertices in $B_i$.  Similarly, we denote
with $\Theta_{i, i+1}$ the set of tracks between buckets $B_i$ and $B_{i+1}$.
For each edge $\{u,v\}$ with $u\in B_i$ and $v \in B_{i+1}$ that is assigned to
a track~$\theta$, we introduce pseudo edges $\{u,\theta\}$ and $\{\theta,v\}$,
and denote with $E_{\theta,i}$ all pseudo edges through~$\theta$ with one endpoint
in~$B_i$.  For each bucket, let $B^*_i = \Theta_{i-1,i}\cup B_i\cup\Theta_{i,i+1}$
be the set of vertices in the bucket plus pseudo vertices for the tracks in
$\Theta_{i-1, i}$ and $\Theta_{i, i+1}$.  Let~$\pi_i$ be the vertex order of
$B^*_i$.

The Sat formulation uses the following variables.
Let $x_{u,v}$ for $u,v\in B^*_i,\ i\in[k]$ equal~\true if $u<_{\pi_i}v$ and~\false,
otherwise. 
Furthermore, let $y_{u,v,s,t}$ for $\{u,v\}, \{s,t\} 
\in E_{\theta,i}\cup E_i\cup E_{\theta',i}, \theta\in\Theta_{i-1,i},
\theta'\in\Theta_{i,i+1},\ \{u,v\}\cap\{s,t\}=\varnothing, i\in[k]$ 
equal~\true if $\{u,v\}$ and $\{s,t\}$ cross and \false, otherwise.

The Sat formulation enforces antisymmetry, i.e., $x_{u,v} = \neg x_{v,u}$,
and transitivity, i.e., $x_{u,v}\land x_{v,w}\Rightarrow x_{u,w}$, for
$u,v,w\in B^*_i,\ i\in[k]$.  We enforce that tracks from the upper gate node
precede all other vertices and tracks from $B^*_i$ in $\pi_i$ by setting
$x_{u,v}=\true$ for $u\in\Theta_{i-1,i},\ v\in B_i\cup\Theta_{i,i+1},\ i\in[k]$.
We require tracks from the lower gate node to form a consecutive subsequence of
$\pi_i$.  Therefore, we enforce $x_{s,w}\Leftrightarrow x_{t,w}$ for $w\in
B_i,\ s,t\in \Theta_{i,i+1},\ i\in[k]$.

For pairs of (pseudo) edges $\{u,v\}$ and $\{s,t\}$, we set up terms like
$x_{u,s}\land x_{s,v}\land x_{v,t} \Rightarrow y_{u,v,s,t}$ for each of the
eight possible crossing patterns.  Finally, our goal is to set as many
$y$-variables to \false as possible.

\section{Experiments}\label{sec:experiments}
We conducted experiments to answer two research questions.
Firstly, do commonly drawn graphs allow for narrow drawings with our framework?
That is, do they have low bucketwidth (or few long edges with respect to a given
bucket capacity)?
Secondly, are the methods presented in~\cref{ssec:arrangement} capable of
finding good bucket arrangements?

\subsubsection*{Setup.}  We used two datasets, the well-known Rome
Graphs~\cite{dglttv-ecfgda-cgta97}, a collection of 11,531 mostly real networks,
and a synthetic dataset of 1,170 $k$-paths (edge-maximal graphs of pathwidth~$k$)
with 10 to 200 vertices and $k\in\{3,5,7\}$.  Pathwidth is an established measure
for how path-like a graph is and we expect them to have good Chunky Chain
visualizations.  \Cref{fig:basics} shows the edge density distribution for the
two datasets.

\begin{figure}[tb]\centering
    \includegraphics[page=3]{bw/bw.pdf}\\
    \begin{subfigure}[t]{0.48\textwidth}
        \includegraphics[page=1]{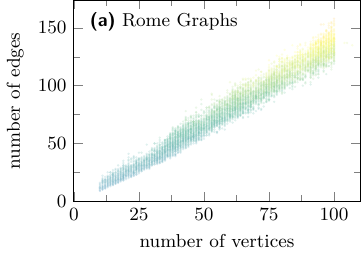}
    \end{subfigure}\hfill
    \begin{subfigure}[t]{0.49\textwidth}
        \includegraphics[page=2]{bw/bw.pdf}
    \end{subfigure}
    \vspace{-1em}
    \caption{Number of vertices and edges for each graph in the two datasets.
    Marker color encodes the bucketwidth of each graph (see scale above).}
    \label{fig:basics}
\end{figure}

We implemented the algorithms described in \cref{ssec:arrangement}.
We used Gurobi optimizer 13.0.1~\cite{gurobi} and HiGHS 1.14.0~\cite{highs,hh-pdrs-mpc18}
for ILP-solving and set a time limit of one hour for both.  We call them
\texttt{Gurobi} and \texttt{HiGHS}, respectively. % in the following sections.
We solve \OBI to obtain a heuristic solution for \BI and use the following 
strategies
to determine the vertex order.  Two are heuristics for bandwidth minimization, 
Cuthill-McKee\cite{cm-rbssm-acm69} (\texttt{cmk}) and node centroid hill
climbing~\cite{lrx-casbmp-hicss04,lrx-fabm-ijait07} (\texttt{nchc}).  The third
one is principal component analysis~\cite{p-pca-pm01} (\texttt{pca})
-- a method for dimensionality reduction~-- applied on the all-pairs
shortest path matrix of~$G$. 
We order the vertices by the first principal component.
The heuristics perform poorly on instances with many long edges.  We therefore
combine them with the meta-heuristic described in \cref{ssec:arrangement} and
call the combined approaches \texttt{cmk*}, \texttt{nchc*}, and \texttt{pca*},
respectively.
Refer to \arxiv{\cref{apx:benchmarks}}\lncs{the full version~\cite{hjsw-cc-arxiv}}
for further details.

\subsubsection*{Results.}  We calculated the bucketwidth of each graph in our
datasets using a variant of the ILP formulation presented in \cref{ssec:arrangement}.
The colors in \cref{fig:basics} show the results. Two patterns become apparent.
Firstly, the stripes of equal color (i.e., equal bucketwidth) are almost
horizontal.  This suggests that the number of edges, rather than the number of
vertices, characterizes more difficult instances.  Secondly, the \tfspaths have
a lower bucketwidth compared to the Rome Graphs considering that they are (on
average) much denser.

\begin{figure}[tb]
    \centering
    \fbox{\small
        \fakemarko{PKdarkblue}\,\crect{PKdarkblue}\,\fakemarko{PKlightblue}\,\crect{PKlightblue}~$c=4$\qquad
        \fakemarkp{PKdarkgreen}\,\crect{PKdarkgreen}\,\fakemarkp{PKlightgreen}\,\crect{PKlightgreen}~$c=8$\qquad
        \fakemarkx{PKdarkorange}\,\crect{PKdarkorange}\,\fakemarkx{PKlightorange}\,\crect{PKlightorange}~$c=12$}
    \begin{tikzpicture}
        \node (main) at (0,0) {\includegraphics[page=1]{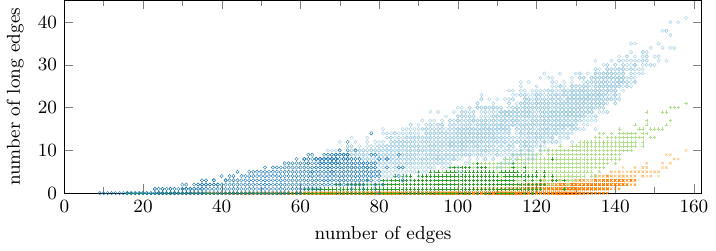}};
        \node (overlay) at (-27mm,9mm) {\includegraphics[page=2]{cml/cml.pdf}};
    \end{tikzpicture}
    \vspace{-2em}
    \caption{Number of long edges~$\ell$ with~$c\in\{4,8,12\}$ for the Rome Graphs.
        Darker shades of the colors represent solutions for which we could prove
        optimality, lighter ones heuristic results.}
    \label{fig:cml-rome}
\end{figure}

We tried to find optimal bucket arrangements for the graphs in our datasets and
$c\in\{4,8,12\}$.  We removed trivial instances where $n\leq 2c$.  Somewhat
expected, our experiments show that it is easier to find such arrangements for
larger capacities.  \Cref{fig:cml-rome} shows that instances that require many
long edges are also harder to solve (\arxiv{see \cref{fig:cml-kpaths} in \cref{apx:sup-figs}}\lncs{refer to the full version~\cite{hjsw-cc-arxiv}}
for the corresponding plot of the \tfspaths).
We can see in \cref{fig:cml-rome} that at least 72\,\% of the Rome Graphs can be
drawn with less than ten long edges using capacity~4.  Note that the median
bucketwidth of the Rome Graphs is seven. % \tim{average 7.48030881332408}.

We assess the quality of our heuristics by measuring the normalized number of
unnecessarily long edges, $(\ell_{\mathrm{heur}}-\ell_{\mathrm{opt}})/m$,
where~$\ell_{\mathrm{heur}}$ is the number of long edges produced by the
respective (heuristic) method and~$\ell_{\mathrm{opt}}$ is the optimum.  We
compared instances where we could not prove optimality to the best result of all
methods.  \Cref{fig:heur-rome-short} shows the results for the Rome Graphs 
(\arxiv{see \cref{fig:heur-rome,fig:heur-kpaths} in \cref{apx:sup-figs}}%
\lncs{see the full version~\cite{hjsw-cc-arxiv}} for further results).

\begin{figure}[tb]
    \centering
    \fbox{\footnotesize
        % \fakeareas{PKlightred}~\texttt{cmk}\quad
        % \fakeareas{PKlightpurple}~\texttt{nchc}\quad
        % \fakeareas{PKlightcyan}~\texttt{pca}\quad
        \fakeareas{PKdarkred}~\texttt{cmk*}\quad
        \fakeareas{PKdarkpurple}~\texttt{nchc*}\quad
        \fakeareas{PKdarkcyan}~\texttt{pca*}\quad
        \fakeareas{PKdarkpink}~\texttt{Gurobi}\quad
        \fakeareas{PKlightbrown}~\texttt{HiGHS}}\vspace{3pt}
    \includegraphics[page=1]{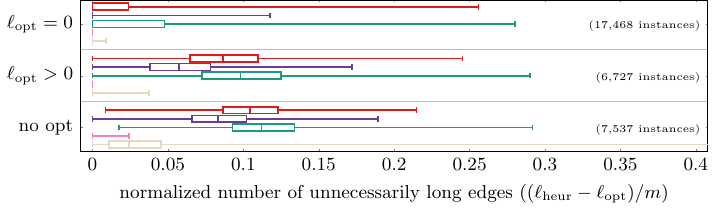}\vspace{-0.3em}
    \caption{Quality of the heuristics on the Rome Graphs measured as percentage
        of unnecessarily long edges $\left((\ell_{\mathrm{heur}}-\ell_{\mathrm{opt}})/m\right)$.
        % The bottommost set of box plots (no opt) contains instances that we could
        % not solve to optimality.  We compare them to the best known result.
        }
    \label{fig:heur-rome-short}
\end{figure}

If $\ell_{\mathrm{opt}}=0$, our heuristics find the optimal solution more often
than not.  Using the meta-heuristic is universally beneficial.  The difference
is particularly pronounced when~$\ell$ is large.  At least on the Rome Graphs,
\texttt{HiGHS} reliably finds (near-)optimal results.  It is, however, very slow
(\arxiv{see \cref{fig:runtime} in \cref{apx:sup-figs}}%
\lncs{see the full version~\cite{hjsw-cc-arxiv}}).  Of the heuristics, \texttt{nchc*}
outperforms the others.  While not quite as good, \texttt{cmk*} is much faster.

\section{Conclusion and Future Work}\label{sec:conclusion}

We introduced a new combinatorial framework for narrow graph drawings.
It admits a family of different drawing styles; as one
such implementation we introduced Chunky Chains.
With our exact and heuristic algorithms for computing bucket
arrangements and our (implemented) pipeline, % with few long edges
we built a foundation for employing this combinatorial framework in
practice and make Chunky Chains available for applications
and further research. 
Some questions remain open:
\begin{itemize}
    \item Solvers have trouble finding good bounds with our ILP formulation. 
        The large number of variables with non-integral values suggests that 
        the model is weak. Can we strengthen it?
    \item The meta-heuristic for \BI removes edges until it finds a bucket
        arrangement with zero long edges.  Sometimes, terminating earlier gives
        a better solution.  Can we identify such cases beforehand?
    \item 
        % Our approach comes with an obvious drawback: although the
        % width of the visualization is reduced, a graph may no longer fit
        % entirely on a screen. Consequently, a global view of the graph
        % can only be achieved through scrolling.
        % Nonetheless, 
        Can user studies show that this drawing style is
        effective on narrow screens compared to other styles?
\end{itemize}

\subsubsection*{Acknowledgements.}  We thank Vasil Alistarov for fruitful
discussions and his prior work~\cite{a-irclm-uniwue25}, Max Kusmitschov for
his thoughts on chord diagram geometry, Joachim Baumeister and David Hock
for their feedback on an early prototype, and our very supportive reviewers.

\pdfbookmark[1]{References}{References}
\bibliography{abbrv,literature}

\arxiv{
\appendix

\section{Proofs Omitted From \texorpdfstring{\cref{sec:combi}}{Combinatorial Results}}
\label{apx:combinatorial}
We use the standard $O^*$-notation to suppress polynomial factors in the $O$-notation.

\numbucketsupperbound*\label{lem:num-buckets-upper-bound*}
\begin{proof}
    \todo[inline, disable]{sw: Improve this bound by not using cliques of the
    same size but cliques of different sizes such that the number of
    three consecutive cliques is $2c+1$.
    
    Make a case distinction depending on $c \mod 3$: If $c$ is a
    multiple of $3$, then make cliques of size $\frac{2c}{3}, \frac{2c}{3}+1, \frac{2c}{3}, \dots$. This ensures that there are $2c+1$ vertices in three consecutive
    buckets.

    If $c \mod 3 = 1$, then use cliques of size $\lfloor\frac{2c}{3}\rfloor+1$.
    
    If $c \mod 3 = 2$, then use $\lfloor\frac{2c}{3}\rfloor+1, \lfloor\frac{2c}{3}\rfloor+1, \lfloor\frac{2c}{3}\rfloor, \dots$.

    Since we have $2c+1$ vertices in three consecutive buckets, we cannot
    compress buckets. We can still argue that a clique cannot be in two different
    buckets without inducing a long edge (except for the first and last bucket).
    $\Rightarrow \Lambda' = \lfloor\frac{n}{\frac{2c+1}{3}}\rfloor = \lfloor\frac{3n}{2c+1}\rfloor$

    We need to show that $\lfloor\frac{3n}{2c+1}\rfloor$ buckets always suffice.
    We do this by directly showing that there are always three consecutive buckets
    that have $2c+1$ vertices, then count how many buckets were there in total.
    }
    We first show that $\lfloor\frac{3n}{2c}\rfloor$ is indeed an
    upper bound on the number of buckets.
    Towards a contradiction, suppose there is a graph~$G$ and a capacity~$c$ 
    where every optimal solution requires at least
    $a := \lfloor\frac{3n}{2c}\rfloor +1$ many buckets.
    Subsequently, every solution has at least $a/3 \ge n/2c$ many (non-overlapping) 
    triplets of consecutive buckets.
    Thus, the average number of vertices in each triplet is at most $2c$, 
    implying that there must be a triplet with at most $2c$ vertices.
    This triplet can be compressed into two buckets by removing the middle
    bucket (and distributing the vertices in the middle bucket) which cannot
    introduce any long edges.
    This is a contradiction to the assumption that every optimal solution
    requires at least~$a$ buckets.

    For any capacity $c$, there exists a family $\mathcal{G}_{c}$ of graphs
    for which every optimal solution requires
    $\Lambda' = \lfloor\frac{3n}{2(c+2)}\rfloor$ many buckets.
    For a given capacity $c$ and a natural number $\Lambda' > 2$,
    the graph $G_{\Lambda'} \in \mathcal{G}_{c}$ consists of ${\Lambda'}$
    cliques $C_1, \dots, C_{\Lambda'}$
    of size $\lfloor 2c/3 \rfloor+ 1$,
    where for each $1\leq i \leq {\Lambda'}-1$, the induced subgraph
    of $C_i \cup C_{i+1}$ is a complete graph.
    
    Obviously, $\Lambda'$ many buckets suffice to construct a solution
    without any long edge.
    Suppose there exists a solution where the vertices of a clique $C_i$,
    $2 \leq i \leq \Lambda'-1$, are contained in two consecutive buckets,
    i.e., there are two vertices $u, v$ of $C_i$ that are in different
    buckets $B_u$, $B_v$. These vertices have at least $2c - 1$ many
    neighbors in common (other than themselves).
    Consequently, at least one neighbor is neither in $B_u$ nor in $B_v$,
    resulting in a long edge.
    Therefore, each $C_i$ with $2\leq i \leq \Lambda'-1$ is contained in
    one bucket.
    Note that $C_1$, $C_{\Lambda'}$ can be distributed among two
    consecutive buckets, respectively. However, there exists at least
    one vertex of $C_1$ ($C_{\Lambda'}$ respectively) that is in a 
    bucket that only contains vertices of $C_1$ ($C_{\Lambda'}$ respectively).
    This follows from the fact that in a bucket that contains some clique
    $C_i$, the entire clique $C_i$ must be in this bucket, the clique
    $C_1$ ($C_{\Lambda'}$ respectively) can only be contained in at most
    two consecutive buckets, and two buckets cannot contain three cliques.
    Therefore, the number of buckets is at least $\Lambda'$.
\end{proof}

\biexactsingleexp*\label{thm:bi-exact-single-exp*}
\begin{proof}
    We use dynamic programming over subsets of vertices of~$G$. Let
    $L$ be a subset of $V(G)$ of size at most $n/2$,
    and let $G[L]$ be the subgraph of~$G$ induced by~$L$.
    We define a table entry $D[L, B]$ to be the minimum number of
    long edges for a bucket arrangement of $G[L]$,
    in which $B \subseteq L$ is the rightmost bucket.

    We compute the table $D$ by using the following base cases.
    We set $D[B, B] = 0$ for all $B \subseteq V$ with $|B| \leq c$.
    This corresponds to the case where all vertices are in a single
    bucket.
    Additionally, we set $D[L, B] = \infty$ for all $|B| > c$ since
    this constitutes an invalid bucket arrangement for any $L$.

    For the recursive case, let $E[U, W]$ denote the set of edges
    in $G$ with one endpoint in $U \subseteq V(G)$ and one endpoint
    in $W\subseteq V(G)$.
    Then, for some $B' \subseteq L \setminus B$ that represents the unique
    adjacent bucket of $B$, the set $E[B, L\setminus (B'\cup B)]$
    describes all long edges that go over~$B'$ (in $G[L]$) in any bucket
    arrangement, where~$B'$ and~$B$ are the last two buckets.
    This allows us to compute $D[L, B]$ as follows.
    \begin{align*}
        D[L, B] = \min_{B' \subseteq L\setminus B}\{D[L\setminus B, B'] + 
        |E[B, L\setminus (B'\cup B)]|\}.
    \end{align*}

    Note that in any (optimal) bucket arrangement, there is a
    bucket $M$ whose removal results in two halves that each have at most
    $n/2$ vertices.
    Based on this observation and~$D$ we compute an optimal solution as follows.
    
    We iterate over all possible subsets $M \subseteq V(G)$
    of size at most~$c$.
    For a given set~$M$, we test all possible subsets~$L$ of vertices
    that lie to the left of $M$; a choice of~$L$ determines the subset
    of vertices~$R$ that lies to the right of~$M$. Hence, we can calculate
    the minimum number of long edges for a fixed~$L$,~$M$, and~$R$ by 
    \begin{align*}
        |E[L, R]| + 
        \min_{B \subseteq L}\{D[L, B] + 
        |E[L\setminus B, M]|\} + 
        \min_{B \subseteq R}\{D[R, B] + 
        |E[R\setminus B, M]|\}.
    \end{align*}
    This corresponds to iterating over two rows in $D$ (one for~$L$ and
    one for~$R$).
    Therefore, we get a total runtime of $O^*(4^n)$ for the computation
    of~$D$, since there are $O(2^n \cdot 2^{n/2})$ entries, and we need 
    $O^*(2^{n/2})$ time to compute one entry.
    We consider the $O(3^n)$ combinations of~$M$ and~$L$ ($M$ and $L$ are
    disjoint).
    For a combination $(M, L)$ we compute the minimum number of long edges
    of this combination in $O^*(2^{n/2})$ time using the formula above.
    Hence, we get a total runtime of  
    $O^*(4^n) + O^*((3\sqrt{2})^n)$ with a space consumption of
    $O((2\sqrt{2})^n\log n)$.
\end{proof}

\pw*\label{lem:pw*}
\begin{proof}
    Take a solution with $\ell$ long edges.
    For two adjacent buckets, create an intermediate bucket that contains the vertices
    of every edge between vertices of the adjacent buckets. This intermediate
    bucket contains at most $2c$ vertices.
    For every long edge $\{u, v\} \in E(G)$, where $u$ is in bucket $i$ and $v$ is
    in bucket $j$ $(i < j)$, choose either $u$ or $v$ and add it into every 
    bucket~$k$ for $i < k < j$.
    The buckets form a valid path decomposition and every bucket has size at most
    $2c + \ell$.
\end{proof}

\fptvc*\label{thm:fpt-vc*}
\begin{proof}
    Let $I = V(G) \setminus C$ be the remaining vertices in $G$.
    Note that~$I$ is an independent set, since~$C$ is a vertex cover.
    The general idea of the algorithm is the following.
    Firstly, we guess a bucket arrangement of the vertices
    in~$C$ in \FPT-time.
    Then, for a fixed bucket arrangement of~$C$, we optimally place
    the vertices of $I$ into buckets in polynomial time by reducing
    this subproblem to a matching problem.
    Out of all bucket arrangements, we choose the one with the fewest long
    edges.
    
    To guess the optimal bucket arrangement of the vertex cover in
    \FPT-time we show that a bucket arrangement of~$C$ with $k = 3|C|$
    buckets suffices to consider.
    Indeed, if we place three (or more) buckets between buckets that
    contain vertices in~$C$, then every vertex~$v$ in~$I$ that is not
    placed immediately adjacent to a bucket that contains a vertex in
    $C$ will only be incident to long edges due to the independence
    of~$I$.
    Thus, we can alternatively add such vertices in a post-processing step using
    additional buckets that are placed at the end of the arrangement.
    As a consequence, we only need to consider buckets that contain vertices in
    $C$ or are immediately adjacent to a bucket with vertices in $C$, implying
    that $3|C|$ buckets suffice.

    Now, assume that a bucket arrangement of $C$ is fixed, i.e., we
    have $B_1, \dots, B_{3|C|}$ buckets where $C$ is divided into buckets
    $B_1, \dots, B_{3|C|}$ conforming to the capacity constraints with
    $c - |B_i|$ free \emph{slots} for each bucket $B_i$
    in which we can place vertices in~$I$.
    We construct an auxiliary bipartite weighted graph~$H$
    whose vertices are comprised of the vertices in $I$, every free slot
    of the bucket arrangement, and a dummy slot-vertex $s_{v}^\mathrm{r}$
    for every vertex in $I$.
    Edges in $H$ are between slots and vertices in~$I$.
    For each edge $\{v, s\}$, where $v \in I$ and $s$ is a free slot,
    we define the weight $w(\{v, s\})$ to be the number of long edges
    incident to~$v$ if we place~$v$ in slot~$s$.
    Since~$N(v)$ is a subset of~$C$, this number is independent of the
    placement of other vertices in~$I$.
    If $s = s_{v}^\mathrm{r}$ we set $w(\{v, s_v^\mathrm{r}\}) = |N(v)|$.
    An optimal placement of the vertices in $I$ given the fixed bucket
    arrangement is now equivalent to finding a minimum weight matching
    in $H$ which can be done in polynomial time.
    Any vertex that is matched with its dummy vertex will be placed in
    an arbitrary free slot or an additional bucket that is appended to
    the bucket arrangement (if necessary).

    In total, we guess $3|C|^{|C|} \in 2^{{O}(|C|\log|C|)}$
    many bucket arrangements. For each bucket arrangement we need polynomial
    time to determine the optimal number of long edges for this
    bucket arrangement.
    Therefore, \BI is \FPT{} with respect to the vertex
    cover number.
\end{proof}

\obi*\label{thm:obi*}
\begin{proof}
    We use dynamic programming.
    Let $D[i, d, f]$ be the minimum number of long edges of a bucket
    arrangement of the induced subgraph of the first $i$ vertices 
    in~$\sigma$ where the last two buckets $B_{k-1}$, $B_k$ have
    a size of $|B_{k-1}| = d$ and $|B_{k}| = f$.
    By definition, an optimal solution can be found by computing
    $\min_{1\leq d,f \leq c} D[n, d, f]$.
    
    We use the base case $D[0, d, f] = 0$ for all $1 \leq d, f\leq c$.
    For $i \geq 1$, we do the following.
    Since we compute an ordered bucket arrangement, $v_i$ has to
    be in the last bucket of any bucket arrangement of the
    first $i$ vertices.
    Thus, any edge~$\{v_i, v_j\}$ in $G[\{v_1, \dots, v_i\}]$
    is a long edge in any ordered bucket arrangement whose last
    two buckets have sizes~$d$ and~$f$, respectively, if~$j$ is
    at most~$i-d-f$.
    Therefore, we can compute the recursive case for 
    $i \geq 1$ and $1 \leq d, f\leq c$ as follows.
    \begin{align*}
        D[i, d, f] = \begin{cases} 
            D[i-1, d, f-1] + 
            |\{\{v_i, v_j\} \in E(G) \mid j \leq i-d-f\}| & \text{if } f > 1, \\
            \displaystyle\min_{1 \leq d' \leq c}\{D[i-1, d', d] + |\{\{v_i, v_j\} \in E(G) \mid j < i-d\}|\} &\text{else.}
        \end{cases}
    \end{align*}
    A naïve implementation of this dynamic program computes
    for each of the $O(nc^2)$ entries with $f>1$ its corresponding value
    in~$O(|N(v)|)$ time, but instead we do the following.

    Since every edge $\{v_i, v_j\}$ is always a long edge if
    $|j-i| > 2c$, we can remove these edges in a preprocessing
    step. As a consequence, $|N(v)| \in O(c)$.
    Now, to process a vertex~$v_i$ for cases with $f > 1$, 
    we order (increasingly) the neighbors 
    $N(v_i) \cap \{v_1, \dots, v_i\}$ of~$v_i$ 
    by their indices in $\sigma$ which can easily be done in
    $O(c\log c)$ time. Let $N_\sigma(v_i)$ be this order.
    We fill the two-dimensional table $D[i, d, f]$ of~$v_i$, by
    iterating over pairs $d, f$ whose sum $d+f$ correspond to    
    increasing values in the range $\{1, \dots, 2c\}$. 
    In other words, we fill the table along the diagonals of
    $D[i, d, f]$ from top-left to bottom-right.
    Since we fill the table for increasing values of $d+f$, we
    can maintain the current number of short edges incident to~$v_i$
    by moving an index~$x$ through the
    sorted neighborhood~$N_\sigma(v_i)$ from right to left.
    Initially, $x=|N_\sigma(v)|$ and points to the rightmost
    neighbor left to~$v_i$ in~$\sigma$.
    For each pair~$d, f$ we test if the index of the vertex $v_j$,
    that~$x$ currently points to, is at least $i-(d+f)$.
    If this is the case, $v_j$ is in one of the last two buckets
    which makes $\{v_i, v_j\}$ a short edge, and we increment~$x$.
    In this way, $|N_\sigma(v_i)| - x$ corresponds to the number
    of long edges incident to~$v_i$ for the current choice of~$d$
    and~$f$.
    Consequently, we can process a vertex~$v_i$ in $O(c^2)$ time
    for $f > 1$, yielding a total runtime of $O(nc^2)$ for this case.
    Moreover, there are $O(nc)$ many entries where $f=1$.
    For each of these entries we can easily compute its value
    in $O(c)$ time. Therefore, the total runtime of this dynamic
    program is $O(nc^2)$.
    Note that we can retrieve an optimal bucket arrangement from
    the table~$D$ using standard techniques.
\end{proof}

\section{Spatial Color Co-occurrence Graphs}\label{apx:colorgraphs}

Spatial color co-occurrence graphs are inspired by Art Palette\footnote{see
\url{https://artsexperiments.withgoogle.com/artpalette/}}, a tool for searching
art collections for artworks that match a certain color combination.  Our use
case, however, is different.  We are interested in colors that often occur in
close proximity to each other.  We visualize them with Chunky Chains (see the
examples shown in \cref{fig:teaser,fig:apx:colorgraphs})\footnote{Find digital
versions of “Starry Night Over the Rhône” at
\url{https://artsandculture.google.com/asset/uQE3XORhSK37Dw}
and of “Still Life with Asphodels” at
\url{https://artsandculture.google.com/asset/rwHGE7FnTOc2wQ}.}.

\begin{figure}[tb]
    \centering
    \includegraphics[width=\textwidth]{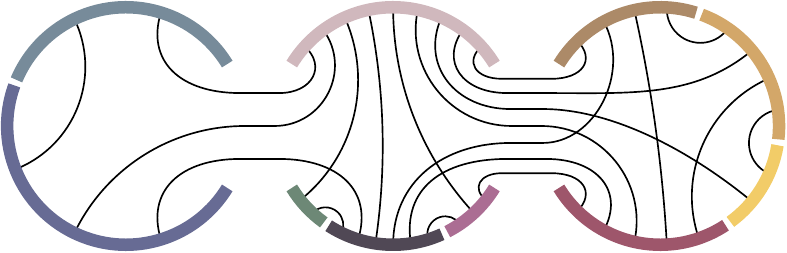}
    \caption{Chunky Chain visualization of spatial color co-occurrences in Henri
        Matisse's “Still Life with Asphodels”.}
    \label{fig:apx:colorgraphs}
\end{figure}

We process the artworks as follows.  We first downscale the image in order to
speed up later steps and to smooth out artifacts like craquelure.  Then, we
apply a quantization filter\footnote{We use the implementation provided with
\url{https://sksamuel.github.io/scrimage/}.} based on the octtree
algorithm~\cite{gp-cqoq-cgi88} to reduce the image to a small set of colors.
We now build a co-occurrence matrix of colors by going through all pairs of
adjacent pixels and count how often color~$a$ is next to color~$b$.  We remove
all but the largest entries such that the remaining entries account for 95\,\%
of the pixel transitions.  Finally, we interpret the matrix as adjacency matrix
and keep the largest connected component.

This type of graph can be used by artists as inspiration to find harmonic
color combinations based on previous work.
Also, such drawings can facilitate master studies of artists (and artworks) as
contrasts, color schemes, the color key (major or minor), and value grouping 
may be easily revealed with our drawings.
The number of long edges can give an indication of how linear the color
transitions in the paintings are.

\section{Computing Bucket Arrangements via ILP}\label{sec:ilp-formulation}
We recall the definition of our variables.
Let $x_{i,v} \in \{0, 1\}$ be a variable for every 
$v\in V(G)$, $i\in[\Lambda]$ that equals~$1$ if and only if~$v$ is in bucket~$i$.
Furthermore, let $y_{i,u,v} \in \{0, 1\}$ for every $\{u,v\}\in E(G)$ and $i\in[\Lambda]$
that equals~$1$ if and only if $\{u,v\}$ is a short edge and~$u$ is in bucket~$i$.
The following ILP formulation computes an optimal bucket arrangement.

\noindent Maximize\quad$\displaystyle\sum_{\substack{\{u,v\}\in E(G)\\i\in[\Lambda]}} y_{i,u,v}$
\begin{align*}
    \text{\llap{\makebox[64.5pt][l]{subject to}}} % ugly af
    \sum_{i=1}^\Lambda x_{i,v} & =1 && \text{for all } v\in V(G) \tag{B1}\label{cons2:b1} \\
    \sum_{v\in V(G)} x_{i,v} &\leq c && \text{for all } i\in[\Lambda] \tag{B2}\label{cons2:b2} \\
    x_{i,u} &\geq y_{i,u,v} && \text{for all } \{u,v\}\in E(G), i\in[\Lambda] \tag{YU}\label{cons2:yu} \\
    x_{i-1,v} + x_{i,v} + x_{i+1,v} &\geq y_{i,u,v} && \text{for all } \{u,v\}\in E(G), i\in\{2,\dots,\Lambda-1\}
        \tag{YV}\label{cons2:yv} \\
    x_{1,v} + x_{2,v} &\geq y_{1,u,v} && \text{for all } \{u,v\}\in E(G) \tag{Y1}\label{cons2:y1} \\
    x_{\Lambda-1,v} + x_{\Lambda,v} &\geq y_{\Lambda,u,v} && \text{for all } \{u,v\}\in E(G)
        \tag{Y$\Lambda$}\label{cons2:yl} \\
    x_{i,w},y_{i,u,v} &\in\{0,1\} && \text{for all } i\in[\Lambda], w\in V(G), \{u,v\}\in E(G)
\end{align*}
Constraints \eqref{cons2:b1} and \eqref{cons2:b2} ensure that each vertex is placed
in exactly one bucket and that each bucket contains at most~$c$ vertices.
\eqref{cons2:yu} and \eqref{cons2:yv} ensure that $y_{i,u,v}$ is~$1$ if and only if
$\{u,v\}$ is short and~$u$ is in bucket~$i$.  \eqref{cons2:y1} and \eqref{cons2:yl} are
special cases of \eqref{cons2:yv} for the first and the last bucket.  The objective
counts the number of short edges (that we want to maximize).

\section{Geometry}\label{apx:geometry}

In \cref{sec:pipeline} we specified the broad shape of a Chunky Chain
without enforcing the exact geometry of every object. Here, we give one
implementation of Chunky Chains
such that the edge complexity remains constant, and we do not
introduce additional crossings.  For aesthetic reasons, we make
all buckets the same size and reserve the same portion of each circle for all
gate nodes (so they all have the same central angle~$\alpha$).  This also
simplifies the construction which works as follows (see \cref{fig:geo}).

\begin{figure}[tb]
    \begin{subfigure}[t]{.33\textwidth}
        \includegraphics[width=0.93\textwidth]{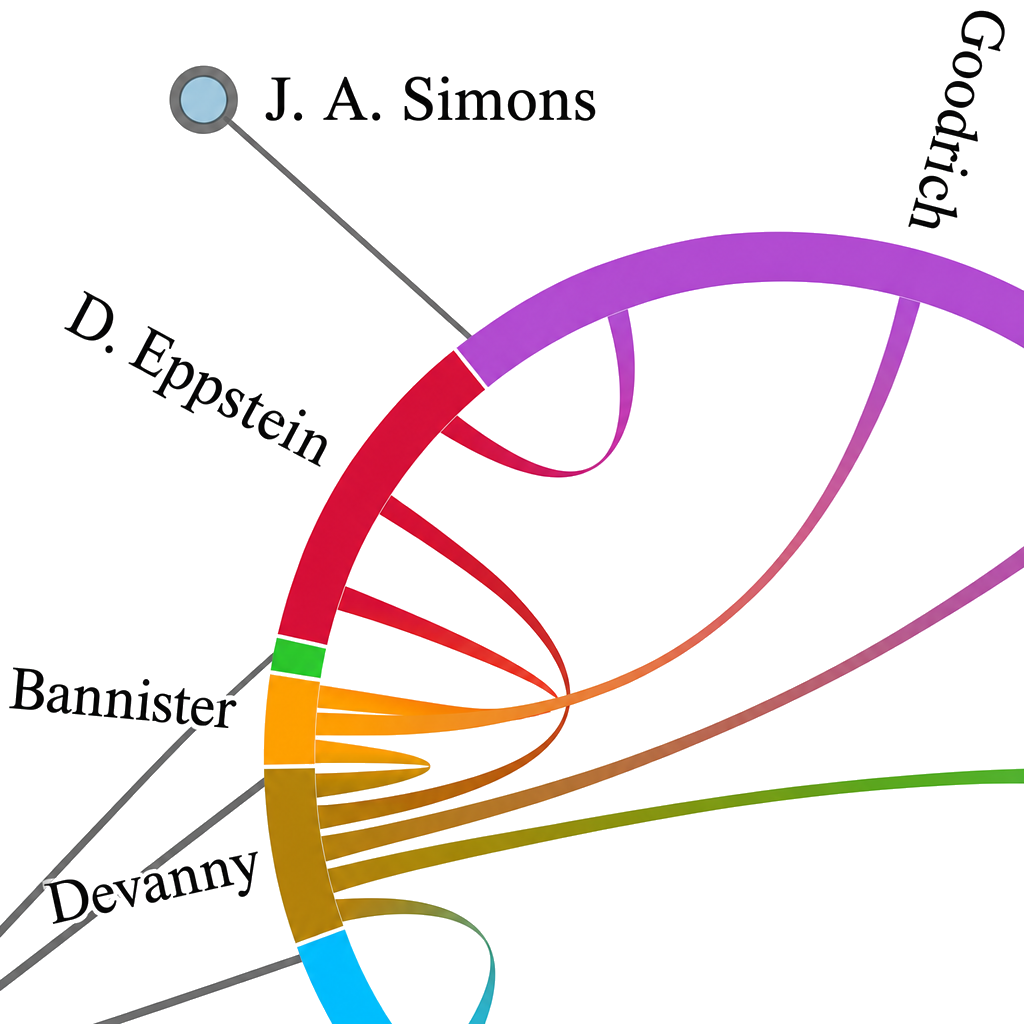}
        \subcaption{Chord diagram from {\cite[Fig.1]{admpt-clhvm-gd19}.}
             % with overlapping chords
        }
        \label{fig:chordlink}
    \end{subfigure}\hfill
    \begin{subfigure}[t]{.26\textwidth}
        \centering
        \includegraphics[page=3]{watzmann.pdf}
        \subcaption{Our chord geometry.}
        \label{fig:geo1}
    \end{subfigure}\hfill
    \begin{subfigure}[t]{.30\textwidth}
        \centering
        \includegraphics[page=4]{watzmann.pdf}
        \subcaption{Some edges need smoothing.}
        \label{fig:geo2}
    \end{subfigure}\hfill
    \caption{(a) overlapping chords introduce unnecessary crossings (see the
        orange chords).  Our chord geometry (b) and (c) uses simple circular
        arcs and effectively avoids unnecessary crossings.}
    \label{fig:geo}
\end{figure}

We call the point where the arc of a node and the curve of one of its adjacent
edges intersect a \emph{port}, and we represent the portion of an edge that is 
between two gate nodes as a vertical line segment.
Given a distance~\dmin, we first find the minimum radius for the buckets such that
(i) any two ports within the same bucket have an arc length of at
least~\dmin,
(ii) there is at least~\dmin space on the circle between any port and any other
node arc,
(iii) the vertical line segments of the edges that connect two gate nodes are at
least~\dmin apart from each other,
(iv) each gate node can have a central angle~$\alpha$ without violating~(i)--(iii).
We draw all buckets with this radius and reserve space for each gate node using~$\alpha$.
% Note that their position and central angle are fixed.
We draw the vertex arcs such that they fill the remaining space on each circle.
The size of each vertex arc is proportional to its vertex degree.

We draw the edges as follows.  Firstly, we determine the position of the
ports of each edge.  For each node arc, we uniformly distribute all its ports on
the arc.  We draw edges with both endpoints in the same bucket as follows.
If the ports of the edge lay exactly antipodal, we draw it as a straight line
segment.  Otherwise, we draw it as a circular arc such that at both ports it is
perpendicular to the bucket's circle (see the black chord in \cref{fig:geo1}).
We split edges that span two adjacent buckets into three parts;  two that both
remain entirely inside one bucket and a third that connects them with a vertical
line segment.  We draw the former parts like edges with both endpoints in the
same bucket.

Note that this construction corresponds to the {\em{Poincaré disk model}}
(hyperbolic geometry), where our circular arcs correspond to lines.
It is well known that two lines in the Poincaré disk model intersect at most once.
Two edges that previously did not intersect must have an even number of crossings
between each other. Hence, using circular arcs cannot introduce new crossings.

With the geometry as described so far, some edges bend when transiting a gate
node (see the red chords in \cref{fig:geo2}).  We replace parts of the edge arcs
with Bézier curves to smooth these bends out (see the blue chords in
\cref{fig:geo2}).  Edges that transit through the center of the gate node are
already smooth by construction.  
However, depending on the actual geometry, smoothing can introduce undesired 
overlaps which requires manual fine-tuning.

\section{Implementation and Benchmarking Environment}\label{apx:benchmarks}

We used Gurobi optimizer 13.0.1\footnote{see \url{https://www.gurobi.com}} and
HiGHS 1.14.0\footnote{see \url{https://highs.dev/}}~\cite{hh-pdrs-mpc18} for
ILP-solving.  For our experiments, we used a set of identical compute nodes each
running Debian 13 with Linux 6.12 on two AMD EPYC™ 9654 96-core processors
alongside 1\,TiB of RAM.  Each solver job had exclusive access to eight cores
and 16\,GiB of RAM (four cores and 8\,GiB of RAM) for Gurobi (HiGHS).  While
Gurobi made use of all eight cores, we could not manage to get HiGHS use more
than two cores.  Both solvers used less than 4\,GiB of RAM with all our instances.

We implemented the heuristics described in \cref{sec:experiments} as well as the
meta-heuristic described in \cref{ssec:arrangement} in the Scala programming
language\footnote{see \url{https://scala-lang.org}}.  The experiments concerned
with the heuristics were conducted on a machine running Fedora 43 with Linux 7.0
on a single AMD Ryzen™ 7 7840HS and 64\,GiB of RAM.  While these processors share
the ZEN4 micro-architecture, a difference in performance is expected.  However,
the results shown in \cref{fig:runtime} span seven orders of magnitude, so we
consider it negligible.

We will publish our implementations as open source software.

\section{Supplemental Figures Omitted From \texorpdfstring{\cref{sec:experiments}}{Experiments}}
\label{apx:sup-figs}

We tried to find optimal bucket arrangements for the graphs in our datasets and
$c\in\{4,8,12\}$.  Somewhat expected, our experiments show that it is easier to
find such arrangements for larger capacities.  \Cref{fig:cml-rome,fig:cml-kpaths}
show that instances that require many long edges are also harder to solve.

Regarding the \tfspaths, recall that we know from \Cref{obs:bi}~\ref{np-hard-pw-tw}
that bounded pathwidth is not necessarily helpful when solving \BI (\cref{lem:pw}
only shows that pathwidth is bounded by the bucket capacity and the number of
long edges (\textit{not} vice versa).  In fact, we argued that we cannot bound
the width of an optimal bucket arrangement with the pathwidth.  Nonetheless, the
pathwidth is an established measure for how path-like a graph is and our
argument is somewhat artificial.  Therefore it is natural to ask if graphs with
bounded pathwidth have good Chunky Chain visualizations in practice. 

In \cref{fig:cml-kpaths}, the number of long edges grows proportional to the
number of edges and, apparently, the graphs are grouped by pathwidth.
\cref{fig:cmlk} shows the normalized number of long edges (i.e., $\ell/m$) for
the \tfspaths.  The box plots show little dispersion which advocates a linear
correlation between the number of long edges and the number of edges, for
constant capacity~$c$ and constant pathwidth.  We conclude that~-- at least
empirically~-- there is a connection between these parameters.

\begin{figure}[htp]
    \centering
    \fbox{\small
        \fakemarko{PKdarkblue}\,\crect{PKdarkblue}\,\fakemarko{PKlightblue}\,\crect{PKlightblue}~$c=4$\qquad
        \fakemarkp{PKdarkgreen}\,\crect{PKdarkgreen}\,\fakemarkp{PKlightgreen}\,\crect{PKlightgreen}~$c=8$\qquad
        \fakemarkx{PKdarkorange}\,\crect{PKdarkorange}\,\fakemarkx{PKlightorange}\,\crect{PKlightorange}~$c=12$}
    \begin{tikzpicture}
        \node (main) at (0,0) {\includegraphics[page=3]{cml/cml.pdf}};
        \node (overlay) at (-27mm,10.5mm) {\includegraphics[page=4]{cml/cml.pdf}};
    \end{tikzpicture}
    \vspace{-2em}
    \caption{Number of long edges~$\ell$ with~$c\in\{4,8,12\}$ for the \tfspaths.
        Darker shades of the colors represent solutions for which we could prove
        optimality, lighter ones represent heuristic results.}
    \label{fig:cml-kpaths}
\end{figure}

\begin{figure}[tb]
    \centering
    \includegraphics{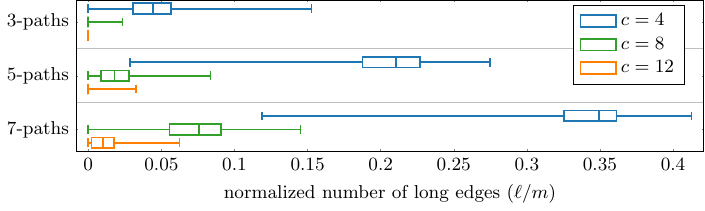}
    \caption{Percentage of long edges ($\ell/m$) with~$c\in\{4,8,12\}$ for the
        \tfspaths.  The lower and upper whiskers mark the minimum and maximum,
        respectively.}
    \label{fig:cmlk}
\end{figure}

We assess the quality of our heuristics by measuring the normalized number of
unnecessarily long edges, $(\ell_{\mathrm{heur}}-\ell_{\mathrm{opt}})/m$,
where~$\ell_{\mathrm{heur}}$ is the number of long edges produced by the
respective (heuristic) method and~$\ell_{\mathrm{opt}}$ is the optimum.  We
compared instances where we could not prove optimality to the best result of all
methods.  \Cref{fig:heur-rome,fig:heur-kpaths} show the results.  If
$\ell_{\mathrm{opt}}=0$, our heuristics find the optimal solution more often
than not.  Using the meta-heuristic is universally beneficial.  The difference
is particularly pronounced when~$\ell$ is large.  At least on the Rome Graphs,
\texttt{HiGHS} reliably finds (near-)optimal results.  It is, however, very slow
(see \cref{fig:runtime}) and often struggles to find good initial solutions for
difficult instances.  Of the heuristics, \texttt{nchc*} outperforms the others.
While not quite as good, \texttt{cmk*} is much faster.  We see no benefit in
using \texttt{pca} over any of the bandwidth heuristics.  When looking at
instances of the \tfspaths that we could not solve to optimality, we notice that
\texttt{nchc*} outperforms the ILP solvers (even the state-of-the-art commercial
solver \texttt{Gurobi}).  We therefore recommend to warm-start ILP solver with a
heuristic solution.

\begin{figure}[tb]
    \centering
    \fbox{\small
        \fakeareas{PKlightred}~\texttt{cmk}\quad
        \fakeareas{PKlightpurple}~\texttt{nchc}\quad
        \fakeareas{PKlightcyan}~\texttt{pca}\quad
        \fakeareas{PKdarkred}~\texttt{cmk*}\quad
        \fakeareas{PKdarkpurple}~\texttt{nchc*}\quad
        \fakeareas{PKdarkcyan}~\texttt{pca*}\quad
        \fakeareas{PKdarkpink}~\texttt{Gurobi}\quad
        \fakeareas{PKlightbrown}~\texttt{HiGHS}}\vspace{3pt}
    \includegraphics[page=2]{heur/heur.pdf}\vspace{-0.3em}
    \caption{Quality of the heuristics on the Rome Graphs measured as percentage
        of unnecessarily long edges $\left((\ell_{\mathrm{heur}}-\ell_{\mathrm{opt}})/m\right)$.
        The bottommost set of box plots (no opt) contains instances that we could
        not solve to optimality.  We compare them to the best known result.}
    \label{fig:heur-rome}
\end{figure}

\begin{figure}[htp]
    \centering
    \fbox{\small
        \fakeareas{PKlightred}~\texttt{cmk}\quad
        \fakeareas{PKlightpurple}~\texttt{nchc}\quad
        \fakeareas{PKlightcyan}~\texttt{pca}\quad
        \fakeareas{PKdarkred}~\texttt{cmk*}\quad
        \fakeareas{PKdarkpurple}~\texttt{nchc*}\quad
        \fakeareas{PKdarkcyan}~\texttt{pca*}\quad
        \fakeareas{PKdarkpink}~\texttt{Gurobi}\quad
        \fakeareas{PKlightbrown}~\texttt{HiGHS}}\vspace{3pt}
    \includegraphics[page=3]{heur/heur.pdf}\vspace{-0.3em}
    \caption{Quality of the heuristics on the \tfspaths measured as percentage
        of unnecessarily long edges $\left((\ell_{\mathrm{heur}}-\ell_{\mathrm{opt}})/m\right)$.
        The bottommost set of box plots (no opt) contains instances that we could
        not solve to optimality.  We compare them to the best known result.}
    \label{fig:heur-kpaths}
\end{figure}

\begin{figure}[tbp]
    \centering
    \fbox{\small
        \fakelines{PKlightred}~\texttt{cmk}\quad
        \fakelines{PKlightpurple}~\texttt{nchc}\quad
        \fakelines{PKlightcyan}~\texttt{pca}\quad
        \fakelines{PKdarkred}~\texttt{cmk*}\quad
        \fakelines{PKdarkpurple}~\texttt{nchc*}\quad
        \fakelines{PKdarkcyan}~\texttt{pca*}\quad
        \fakelined{PKdarkpink}~\texttt{Gurobi}\quad
        \fakelined{PKlightbrown}~\texttt{HiGHS}}\vspace{3pt}
    \begin{subfigure}[t]{0.5\textwidth}
        \includegraphics[page=1]{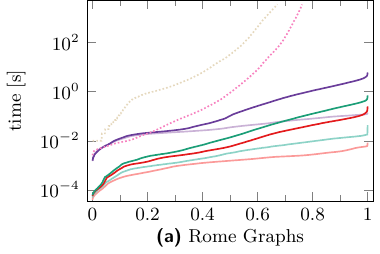}
    \end{subfigure}
    \hfill
    \begin{subfigure}[t]{0.48\textwidth}
        \includegraphics[page=2]{runtime/runtime.pdf}
    \end{subfigure}\vspace{-0.3em}
    \caption{Running times per share of the dataset.  The y-axis is logarithmic.}
    \label{fig:runtime}
\end{figure}

\section{Original Version of \texorpdfstring{\cref{fig:chordlink}}{[Angori et al., Fig.1]}}
\label{apx:chordlink-original}

\Cref{fig:chordlink} has been upscaled using generative AI.  The original version
taken from {\cite[Fig.1]{admpt-clhvm-gd19}} has been reproduced as~\cref{fig:chordlink-original}
for comparison.

\begin{figure}[htbp]
    \centering
    \includegraphics[width=0.307\textwidth]{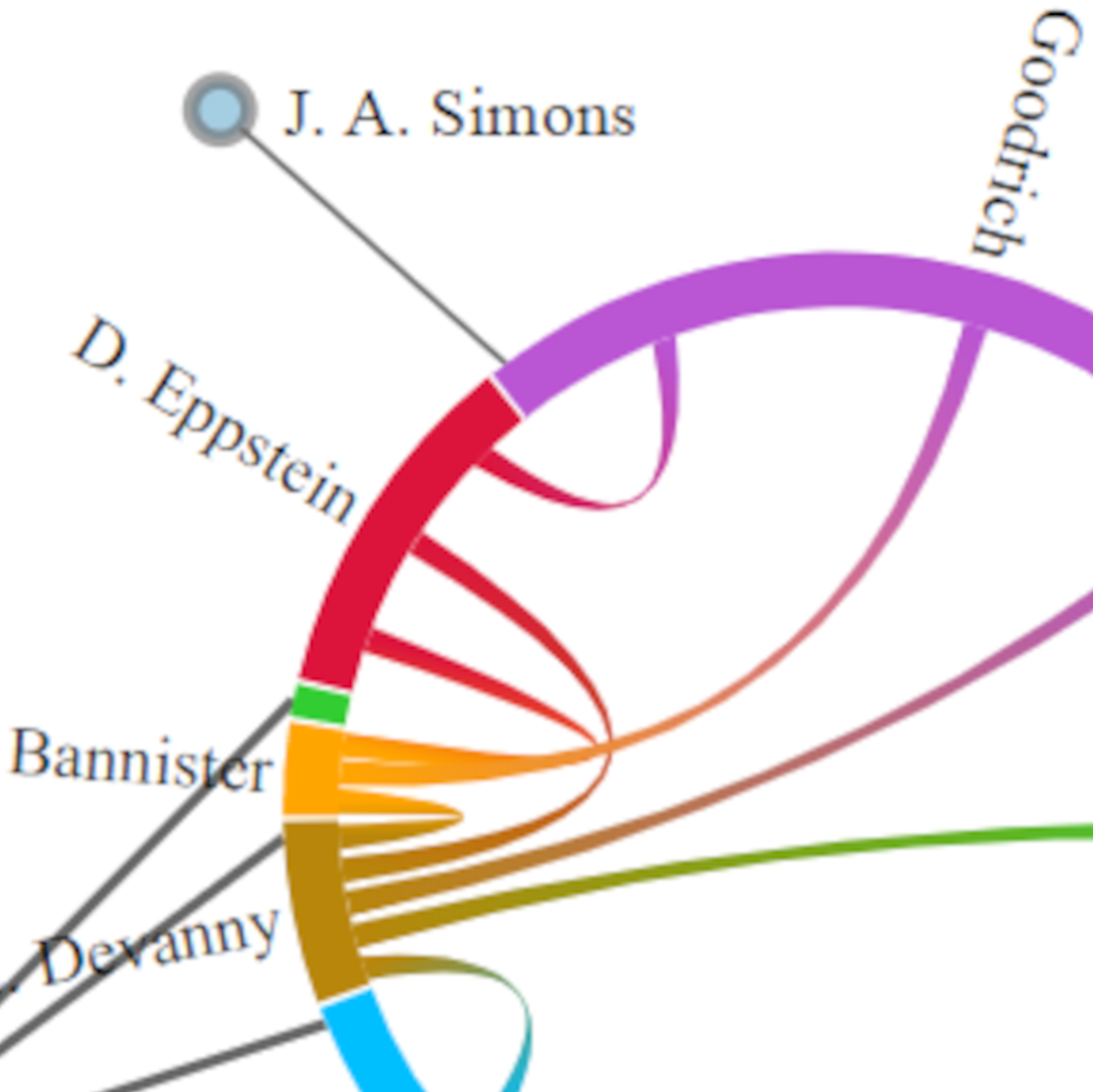}
    \caption{Original version of~\cref{fig:chordlink} taken from
        {\cite[Fig.1]{admpt-clhvm-gd19}}.}
    \label{fig:chordlink-original}
\end{figure}
}
\end{document}